\documentclass[3p,sort&compress]{elsarticle}
\usepackage{amsmath}
\usepackage{color}
\usepackage{amsthm}
\usepackage{amssymb}
\usepackage{titlesec}
\usepackage[thinlines,thiklines]{easybmat}

\usepackage{graphicx}
\usepackage{epsfig}

\theoremstyle{definition}\newtheorem{definition}{Definition}
\theoremstyle{plain}\newtheorem{theorem}{Theorem}
\theoremstyle{definition}\newtheorem{remark}{Remark}
\theoremstyle{definition}
\theoremstyle{plain}
\theoremstyle{plain}\newtheorem{corollary}{Corollary}
\theoremstyle{plain}\newtheorem{lemma}{Lemma}
\theoremstyle{plain}

\begin{document}
\begin{frontmatter}
\title{Potential of quantum finite automata with exact acceptance}

\author{Jozef Gruska$^{1}$}
\author{Daowen Qiu$^{2}$ }
\author{Shenggen Zheng$^{1,}$\corref{one}}

 \cortext[one]{Corresponding author.\\ \indent{\it E-mail addresses:} zhengshenggen@gmail.com (S. Zheng), gruska@fi.muni.cz (J. Gruska),  issqdw@mail.sysu.edu.cn (D. Qiu).}

\address{

  $^{1}$Faculty of Informatics, Masaryk University, Brno 60200, Czech Republic\\

  $^2$Department of Computer Science, Sun Yat-sen University,
Guangzhou 510006,
  China\\
}

\begin{abstract}

The potential of the exact quantum information processing is an interesting, important and intriguing issue. For examples, it has been believed that quantum tools can provide significant, that is larger than polynomial,  advantages  in the case of exact quantum computation only, or mainly, for  problems with very special structures. We will show that this is not the case.

In this paper the potential of quantum finite automata producing outcomes not only with a (high) probability, but with certainty (so called exactly) is explored in the context of their uses for solving promise problems and with respect to the size of automata. It is shown that for solving particular classes $\{A^n\}_{n=1}^{\infty}$ of promise problems, even those without some very special structure, that  succinctness of the exact quantum finite automata under consideration, with respect to the number of (basis) states,  can be very small (and constant) though it grows proportional to $n$ in the case  deterministic finite automata (DFAs) of the same power are used. This is here demonstrated also for the case that the component languages of the promise problems solvable by DFAs are non-regular. The method used can be applied in finding more exact quantum finite automata or quantum algorithms for other promise problems.
\end{abstract}


\end{frontmatter}

\section{Introduction}\label{sec-intr}

 Because of the probabilistic nature of quantum measurement and a need to have
a quantum measurement involved at any quantum computation of classical problems,
it used to be assumed that for having a significantly superior power of
quantum computation we need to pay by having its correct outcomes only with
some (high) probability.
In other words, it was assumed that we cannot have exponential or even larger increases of quantum power in the case of the outcomes with probability 1 (so called exact computation) are
required.

After the discovery of the exact algorithm for Deutsch's problem \cite{DJ92} and Simon's
problem \cite{Sim97}, it started to be of large interest to get a deeper insight into
the potential of exact quantum computation and communication \cite{BH97,Buh98,BCWZ99,CEMM98,Kla00,MNYW05}. Especially in the recent years, the exact quantum computation of total functions and partial functions (promise problems) has been studied  \cite{AmYa11,Amb13,AISJ13,AGZ14,GDZ14,MJM11,MNYW05,Nak14,RY14,Zhg13} and also  a new
hypothesis has appeared. Namely, that exact, and significantly more powerful than in classical case, quantum computation, is possible only/mainly in case of the  problems having some very special structures.

There are many ways to explore the power of  exact quantum information processing. One of the most basic  approaches is on the level of quantum finite
automata (of various types) and with respect to their sizes in comparison with
the exact classical (deterministic)  finite automata.

The acceptance of languages has been used to be one of the main ways to get an insight into the power of various models of finite automata. Recently, a more general approach,
through considerations of so called promise problems $A=(A_{yes}, A_{no})$, where $A_{yes}$ and $A_{no}$ are disjoint languages, has turned out to be of interest for several reasons. One of them is that this also allows to distinguish ``acceptance" and ``solvability" as two different modes of actions. Second one is that this allows to consider for finite automata also the cases that $A_{yes}$ and $ A_{no}$ are not regular.

The
promise problems, studied in the exact quantum computing before, used to have two very closely related
parameters, as $n$ and $n/2$ or so \cite{AmYa11,Buh98,DJ92}.
In this paper we demonstrate that  for promise problems with no very special structure, i.e., the problems that we  deal with will only have very
loosely related parameters as $N$ and $l$ for any $l<N$.  The quantum exact mode can be, comparing with the classical case of deterministic finite automata (DFAs), even more than exponentially succinct.  This will be demonstrated for the case of a unary (alphabet) promise problem  and also for the case of a binary promise problem $A$ whose components $A_{yes}, A_{no}$ are
even non-regular languages.

Klauck \cite{Kla00} proved, for any regular language $L$,  that the state complexity  of an exact one-way quantum finite automaton accepting $L$ is not less than the state complexity of an equivalent DFA.
However, the situation is very different for some promise problems  \cite{AmYa11,GDZ14,Zhg13}.

A {\it promise problem} over an alphabet $\Sigma$ is a pair $A = (A_{yes}, A_{no})$, where $A_{yes}\subset \Sigma^*$ and $A_{no}\subset \Sigma^*$
are disjoint \cite{Gh06}.
Let $A_{yes}^k=\{a^{i\cdot2^{k+1}}\,|\, i>0 \}$ and  $A_{no}^k=\{a^{i\cdot2^{k+1}+2^k}\,|\, i\geq 0 \}$, where $k$ is any  positive integer.
Ambainis and Yakary{\i}lmaz \cite{AmYa11} proved that any  promise problem $A^k=(A_{yes}^k,A_{no}^k)$
 can be solved exactly by a  measure-once one-way quantum finite automaton (MOQFA)  with two quantum basis states, whereas
the sizes of  the corresponding DFAs are at least $2^{k+1}$.

 In order to get a quantum speed-up (space efficiency) larger than polynomial in the exact quantum computing, it used to be believed that the problem must have a  special structure. For example, in case of  the
  Deutsch-Jozsa promise problem \cite{DJ92}, the distributed  Deutsch-Jozsa promise problem \cite{Buh98}  and the promise problem studied in \cite{AmYa11} could be of such a type\footnote{
  In the Deutsch-Jozsa promise problem, the task is to determine whether an input $x\in\{0,1\}^n$ has the Hamming weight $0$ or $n$ (this is constant) or $\frac{n}{2}$ (this is balanced).
  In the distributed {\em  Deutsch-Jozsa promise problem}, two parties are to determine whether their respective strings $x,y\in\{0,1\}^n$ have the Hamming distance $0$  or $\frac{n}{2}$.     The promise problem studied by Ambainis and Yakary{\i}lmaz \cite{AmYa11} has two parameters, i.e. one is $n$ in $A_{yes}^k$ and the other one $\frac{n}{2}$ in $A_{no}^k$. We say that a promise problem has very special structure if its parameters have very special relationship. }.

We show in this paper that similar results hold also for some of those promise problems where parameters in  $A_{yes}$ and  $A_{no}$ are loosely related.
 To show that we  consider a unary promise problem  $A^{N,\,l}=(A_{yes}^{N,\,l},A_{no}^{N,\,l})$ with  $A_{yes}^{N,\,l}=\{a^{iN}\,|\,\ i\geq 0\}$ and $A_{no}^{N,\,l}=\{a^{iN+l}\,|\,\ i\geq 0\}$, where $N$ and $l$ are  fixed positive integers such that $0< l<N$. We prove that the promise problem $A^{N,\,l}$ can be solved exactly by a three quantum  basis states MOQFA. Note that if we choose $N= 2^{k+1}$ and $l=N/2=2^k$, then $A^{N,\,l}$ is  the promise problem   studied in \cite{AmYa11}.
We  also determine the state complexity of DFA solving the above promise problems.  We prove that for any fixed  $N$ and $l$,  the minimal DFA solving the promise problem $A^{N,\,l}$ has $d$ states,  where $d$ is the  smallest integer such that $d\mid N$ and $d\nmid l$. If $N$ is a prime, then the minimal DFA solving the promise problem $A^{N,\,l}$ has $N$ states.
 Finally, we consider even a more general promise problem. Namely, $A^{N,r_1,r_2}=(A_{yes}^{N,r_1},A_{no}^{N,r_2})$ with  $A_{yes}^{N,r_1}=\{a^{n}\,|\,\ n \equiv r_1\ {\it mod}\  N\}$ and $A_{no}^{N,r_2}=\{a^{n}\,|\,\ n \equiv r_2\ {\it mod}\  N\}$, where $N$, $r_1$ and $r_2$ are  fixed positive integers such that $r_1\not\equiv r_2\ {\it mod}\  N$. (If we choose $r_1=0$ and $r_2=l$, then $A^{N,r_1,r_2}$ is  the promise problem $A^{N,\,l}$ given above.)
We  prove also that the promise problem $A^{N,r_1,r_2}$ can be solved exactly by an MOQFA with three quantum basis states.  Let  $l=(r_2-r_1) \ {\it mod}\  N$. We show also that the size of the minimal  DFA for the promise problem $A^{N,r_1,r_2}$ is the same as the size of  the minimal DFA  for the promise problem $A^{N,\,l}$.

We consider afterwards  two   binary promise problems. For any positive integer $l$, let  $B_{yes}^{l}=\{a^ib^{i}\mid i\geq 0\}$ and $B_{no}^{l}=\{a^ib^{i+l}\mid i\geq 0\}$. It is easy to see  that both $B_{yes}^{l}$ and $B_{no}^{l}$ are {\em nonregular languages}. We  prove that the promise problem $B^{l}=(B_{yes}^{l},B_{no}^{l})$ can be solved by an exact MOQFA and also by a DFA.  In particular, we  prove that the promise problem $B^{l}$ can be solved exactly by a 2  quantum basis states MOQFA ${\cal M}_{l}$, whereas the corresponding minimal DFA solving this promise problem has $d$ states, where $d$ is the smallest integer such that $d\nmid l$. Furthermore, we consider a more general binary promise problem with $B_{yes}^{N,\,l}=\{a^ib^{i}\mid i\geq 0\}$ and $B_{no}^{N,\,l}=\{a^ib^{i+jN+l}\mid i,j\geq 0\}$, where $N$ and $l$ are fixed nonnegative integers such that  $0<l<N$. We  prove that each promise problem $B^{N,\,l}=(B_{yes}^{N,\,l},B_{no}^{N,\,l})$ can be solved exactly by a 3 quantum  basis states MOQFA ${\cal M}_{N,\,l}$, whereas the corresponding minimal DFA  has  $d$ states, where $d$ is the  smallest integer such that $d\mid N$ and $d\nmid l$. Note that if $N$ is a prime, then $d=N$.

 The paper is structured as follows. In Section 2 some of the  required basic concepts and notations are introduced and  the models   used are
described in  detail.   The results on unary promise problems are given in Section 3. The results on binary promise problems are given in Section 4.   Section 5 contains conclusion and discussion.

\section{Preliminaries}

Quantum finite automata were first introduced by  Kondacs and Watrous \cite{Kon97} and by Moore and Crutchfields \cite{Moo97}. Since that time they have been explored in \cite{Amb02,Bro99,Hir10,LiQiu09,Yak11,Yak10,ZhgQiu112,Zhg12,ZhgQiu11,Zhg13a} and in other papers. The state complexity of quantum finite automata is one of the  interesting topics in studying the power of quantum finite automata.  In the past twenty years, state complexity of several variants of quantum finite automata has been intensively studied  \cite{Amb98,AmbNay02,Amb09,AmYa11,BMP03,Ber05,BMP06,Fre09,Gru00,GDZ14,Qiu14,Zhg12,Zhg13,Zhg14}.

We now recall  some basic concepts    and  notations concerning  finite automata. For  quantum information processing and more on  finite automata we refer the reader to \cite{Gru99,Hop,Nie00,Qiu12}.

\begin{definition}
A {\it deterministic finite automaton} (DFA) $\mathcal{A}$ is specified by a 5-tuple
\begin{equation}
\mathcal{A}=(S,\Sigma,\delta,s_0, F),
\end{equation}
where:
\begin{itemize}
\item $S$ is a finite set of classical states;
\item $\Sigma$ is a finite set of input symbols;
\item $s_{0}\in S$ is the initial state of the automaton;
\item $F\subset S$ is the set of accepting states;
\item $\delta$ is a transition function:
\begin{equation}
\delta:S\times\Sigma \rightarrow S.
\end{equation}
\end{itemize}
\end{definition}

For any $x\in\Sigma^*$ and $\sigma\in \Sigma$, we define
\begin{equation}
    \widehat{\delta}(s,x\sigma)=\widehat{\delta}(  \widehat{\delta}(s,x), \sigma)
\end{equation}
and if $x$ is the empty string, then
 \begin{equation}
    \widehat{\delta}(s,x\sigma)=\delta(s,\sigma).
\end{equation}
 The automaton $\mathcal{A}$
accepts the string $x$ if $\widehat{\delta}(s_0,x)\in F$, otherwise rejects.

 A promise problem $A = (A_{yes}, A_{no})$ is said to be solved  by  a DFA  ${\cal A}$ if
\begin{enumerate}
\item[1.] $\forall x\in A_{yes}$, $\widehat{\delta}(s_0,x)\in F$, and
\item[2.] $\forall x\in  A_{no}$, $\widehat{\delta}(s_0,x)\not\in F$.
\end{enumerate}

\begin{definition}A measure-once quantum finite automaton (MOQFA) ${\cal M}$ is specified by a 5-tuple
 \begin{equation}
{\cal M}=(Q,\Sigma,\{U_{\sigma}\,|\, \sigma\in\Sigma' \},|{0}\rangle,Q_a)
\end{equation}
where:

\begin{itemize}
\item  $Q$ is a finite set of orthonormal quantum (basis) states, i.e. $\{|i\rangle\mid 0\leq i< |Q|\}$;

 \item $\Sigma$ is a finite alphabet of input symbols, which does not contain end-markers.
$\Sigma'=\Sigma\cup \{|\hspace{-1.5mm}c,\$\}$ {\em (}where $|\hspace{-1.5mm}c$ will be used as the left end-marker and $\$$ as the right end-marker{\em )};
\item $|0\rangle\in Q$ is the initial quantum state;

\item $Q_a \subseteq Q$ denotes the set of
accepting basis states;

\item $U_{\sigma}$'s are unitary operators for $\sigma\in\Sigma'$.
\end{itemize}

The quantum state space of this model is a $|Q|$-dimensional Hilbert space denoted ${\cal H}_Q$. Each quantum basis state $|i\rangle$  in ${\cal H}_Q$ is represented by a column vector with the $(i+1)$th entry being $1$ and the other entries being $0$.
With this notational convenience we can describe the above model as follows:
\begin{enumerate}
  \item  The initial state $|0\rangle$ is represented as $|q_0\rangle=(1,\overbrace{0,\cdots,0}^{|Q|-1})^\mathrm{T}$.
  \item  The accepting set $Q_a$ corresponds to the projective operator $P_{acc}=\sum_{|i\rangle\in Q_a}|i\rangle\langle i|$.
\end{enumerate}

The computation of an MOQFA ${\cal M}$ on an input string
$x=\sigma_{1}\sigma_{2}\cdots\sigma_{n}\in\Sigma^{*}$ is as
follows:  ${\cal M}$ reads the input string from the left end-marker to the right end-marker,  symbol by symbol, and  the unitary matrices $U_{|\hspace{-1mm}c}, U_{\sigma_1},U_{\sigma_2},\cdots,U_{\sigma_n},U_{\$}$ are applied, one by one, on the current state, starting with  $|0\rangle$ as the initial state. Finally, the projective
measurement $\{P_{acc}, I-P_{acc}\}$ is performed on the final state, in order to accept or reject
the input. Therefore, for the input string $x$,  ${\cal M}$ has the accepting probability given by
\begin{equation}
Pr[{\cal M}\ \mbox{accepts}\  x] =\|P_{acc}U_{\$}U_{\sigma_n}\cdots U_{\sigma_2}U_{\sigma_1}U_{|\hspace{-1mm}c}|0\rangle\|^2.
\end{equation}
\end{definition}

 A promise problem $A = (A_{yes}, A_{no})$ is solved exactly by  an MOQFA ${\cal M}$ if
\begin{enumerate}
\item[1.] $\forall x\in A_{yes}$, $Pr[{\cal M}\  \mbox{accepts}\  x]=1$, and
\item[2.] $\forall x\in  A_{no}$, $Pr[{\cal M}\ \mbox{accepts}\  x]=0$.
\end{enumerate}

\section{Unary promise problems}\label{s-unary}

In this section we prove that a specific unary promise problem can be solved exactly by  an  MOQFA with a small (constant) number of quantum states, but any DFA solving the same promise problem has to have much larger  number of states.   To show that we make use of a special technique introduced by Ambainis \cite{Amb13}.  Based on Ambainis' idea, we introduce a definition and deal with a lemma first.

\begin{definition}
Let $p\in[-1,0]$. A quantum machine   ${\cal M}$ with initial state $|0\rangle$   $p$-solves promise problem  $A=(A_{yes},A_{no})$ if,
\begin{enumerate}
\item $\forall x\in A_{yes}$, $U_x|0\rangle=|0\rangle$,
\item $\forall x\in  A_{no}$, $U_x|0\rangle=p|0\rangle+\sqrt{1-p^2}|\psi_x\rangle$
\end{enumerate}
where   $|\psi_x\rangle\bot|0\rangle$ and the unitary operator $U_x$ is the action of the machine ${\cal M}$ corresponding to the input $x$.
\end{definition}
Hence  if $x\in A_{yes}$, then the amplitude of $|0\rangle$ in the quantum state $U_x|0\rangle$ is 1.  If $x\in  A_{no}$, then the amplitude of $|0\rangle$ in the quantum state $U_x|0\rangle$ is  $p$.

\begin{equation}\label{e-7}
    \left(
      \begin{array}{ccc}
        1  \\
        0  \\
        \vdots\\
        0 \\
      \end{array}
    \right)
      \begin{array}{ccc}
        {}  \\
        U_x  \\
        \longrightarrow \\
        x\in A_{yes}\\
      \end{array}
      \left(
      \begin{array}{ccc}
        1  \\
        0  \\
        \vdots\\
        0 \\
      \end{array}
    \right)
    \mbox{\ and\ } \left(
      \begin{array}{ccc}
        1  \\
        0  \\
        \vdots\\
        0 \\
      \end{array}
    \right)
      \begin{array}{ccc}
        {}  \\
        U_x  \\
        \longrightarrow \\
        x\in A_{no}\\
      \end{array}
      \left(
      \begin{array}{ccc}
        p  \\
        *  \\
        \vdots\\
        * \\
      \end{array}
    \right).
\end{equation}

Now we  add one entry (as the first one) to the above vectors,  then start with a new initial quantum state, and change the unitary $U_x$ as followings:
\begin{equation}\label{e-8}
    \left(
      \begin{array}{ccc}
        \alpha  \\
        \beta  \\
         0 \\
        \vdots\\
        0 \\
      \end{array}
    \right)
      \begin{array}{ccc}
        {}  \\
        U'_x  \\
        \longrightarrow \\
        x\in A_{yes}\\
         {}  \\
      \end{array}
       \left(
      \begin{array}{ccc}
        \alpha  \\
        \beta  \\
         0 \\
        \vdots\\
        0 \\
      \end{array}
    \right)
    \mbox{\ and\ } \left(
      \begin{array}{ccc}
        \alpha  \\
        \beta  \\
         0 \\
        \vdots\\
        0 \\
      \end{array}
    \right)
      \begin{array}{ccc}
        {}  \\
        U'_x  \\
        \longrightarrow \\
        x\in A_{no}\\
         {}  \\
      \end{array}
      \left(
      \begin{array}{ccc}
        \alpha  \\
        p\beta  \\
         * \\
        \vdots\\
        * \\
      \end{array}
    \right),
\end{equation}
where
\begin{equation}
U'_x=\left(
      \begin{array}{ccc}
        1  \ \ \ \  &  \mathbf{0} \\
        \mathbf{0}  \ \ \ \  &  U_x
      \end{array}
    \right),
\end{equation}
 $\alpha$, $\beta$ are real numbers (for the moment arbitrary), and $*$'s are some values that we do not need to specify  exactly.

In order to have proper quantum states $(\alpha,\beta,0,\ldots,0)^T\bot(\alpha,p\beta,*,\ldots,*)^T$, the following relations have to hold:
\begin{equation}\label{E-eqs}
\left\{
      \begin{array}{ccc}
        \alpha^2+\beta^2=1 \\
         \alpha^2+p\beta^2 =0
      \end{array}
   \right..
\end{equation}
Since $p\leq 0$,  the above equations have a solution  as
\begin{equation}\label{e-p}
\left\{
      \begin{array}{ccc}
        \alpha=\sqrt{\frac{-p}{1-p}}\\
        \beta=\sqrt{\frac{1}{1-p}}.
      \end{array}
   \right.
\end{equation}
 Therefore, we have the following lemma.

\begin{lemma}\label{Lm-1}
If quantum machine ${\cal M}$ $p$-solves a promise problem $A$ where $p\in [-1,0]$, then there exists a quantum machine ${\cal M}'$ that solves the promise problem $A$ exactly.
\end{lemma}
\begin{proof}
Since the quantum machine ${\cal M}$ $p$-solves the promise problem $A$,  we have  $U_x|0\rangle=|0\rangle$ for any $x\in A_{yes}$ and
 $U_x|0\rangle=p|0\rangle+\sqrt{1-p^2}|\psi_x\rangle$ for $ x\in  A_{no}$, where $|\psi_x\rangle\bot|0\rangle$.
With respect to the analysis above,  we  can have   a new quantum machine ${\cal M}'$   such that $U'_x|0'\rangle=|0'\rangle$ for any $x\in A_{yes}$ and
 $U'_x|0'\rangle=(\alpha,p\beta,*,\ldots,*)^T$ for $ x\in  A_{no}$,  where $|0'\rangle$ is the initial quantum state of the new machine,   $U'_x$ is the action of  the new machine corresponding to the input $x$, and $|0'\rangle\bot (\alpha,p\beta,*,\ldots,*)^T$.
 By choosing  the measurement operator $\{|0'\rangle\langle0'|,\  I-|0'\rangle\langle0'|\}$,  the quantum machine ${\cal M}'$ can solve the promise problem $A$ exactly.

\end{proof}

Now we will deal with the promise problems $A^{N,\,l}$ for the case   $\frac{N}{4}\leq l\leq \frac{3N}{4}$.
\begin{lemma}\label{Th-1}
For any fixed positive integers $N$ and $l$ such that $\frac{N}{4}\leq l\leq \frac{3N}{4}$, the promise problem $A^{N,\,l}$ can be solved exactly by a 3 quantum basis states MOQFA ${\cal M}_{N,\,l}$.
\end{lemma}
\begin{proof}
If we choose $\theta=\frac{2\pi}{N}$ and design an MOQFA ${\cal A}$ with two quantum states\footnote{The two quantum states are $|0\rangle=(1,0)^T$ and $|1\rangle=(0,1)^T $.} as the one in Theorem 1 in \cite{AmYa11}, then it is easy to see that  ${\cal A}$ $p$-solves promise problem $A^{N,\,l}$, where $p=\cos l\theta$. Since $\frac{N}{4}\leq l\leq \frac{3N}{4}$, we have $p=\cos l\theta\leq 0$. We can now design an MOQFA  ${\cal M}_{N,\,l}$ to solve the promise problem $A^{N,\,l}$ exactly according to the method specified before  Lemma \ref{Lm-1}.

According to Equalities (\ref{e-p}),
we have $\alpha=\sqrt{\frac{-p}{1-p}}=\sqrt{\frac{-\cos l\theta}{1-\cos l\theta}}$ and $\beta=\sqrt{\frac{1}{1-p}}=\sqrt{\frac{1}{1-\cos l\theta}}$.
Now let ${\cal M}_{N,\,l}=(Q,\Sigma,\{U_{\sigma}\,|\, \sigma\in\Sigma' \},|{0}\rangle,Q_a)$, where $Q=\{|0\rangle,|1\rangle,|2\rangle\}$, $Q_a=\{|0\rangle\}$,
\begin{equation}
    U_{|\hspace{-1mm}c}=\left(
      \begin{array}{ccc}
        \alpha \ & -\beta \ & 0  \\
        \beta \ & \alpha \ & 0   \\
        0 \ &  0  \  &1\\
      \end{array}
    \right),\
    U_{a}=\left(
      \begin{array}{ccc}
        1 & 0 & 0  \\
        0 & \cos \theta & -\sin \theta   \\
        0 &  \sin \theta  &\cos \theta\\
      \end{array}
    \right) \mbox{\ and\ }
    U_{\$}=U_{|\hspace{-1mm}c}^{-1}=\left(
      \begin{array}{ccc}
        \alpha \ & \beta \ & 0  \\
        -\beta \ & \alpha \ & 0   \\
        0 \ &  0  \  &1\\
      \end{array}
    \right).
\end{equation}

$U_{|\hspace{-1mm}c}$ is the unitary  matrix that  changes the initial quantum state $(1,0,0)^T$ to $(\alpha,\beta,0)^T$ according the processes shown in  Equality (\ref{e-7}) and (\ref{e-8}).

If the input  $x\in A^{N,\,l}_{yes}$, then the quantum state before the measurement is
\begin{equation}
    |q\rangle=U_{\$}(U_a)^{iN}U_{|\hspace{-1mm}c}|0\rangle=U_{\$}IU_{|\hspace{-1mm}c}|0\rangle=|0\rangle,
\end{equation}
because $(U_a)^{iN}=I$.

If the input  $x\in A^{N,\,l}_{no}$, then the quantum state before the measurement is
\begin{align}
 |q\rangle&=U_{\$}(U_a)^{iN+l}U_{|\hspace{-1mm}c}|0\rangle=U_{\$}(U_a)^{l}U_{|\hspace{-1mm}c}|0\rangle\\ \label{e1}
 &=U_{\$}(U_a)^{l}\left( \alpha|0\rangle+ \beta|1\rangle\right)\\
 &=U_{\$}\left( \alpha|0\rangle+ \cos l\theta\cdot\beta|1\rangle+\gamma|2\rangle\right)\\
 &=U_{\$}\left( \alpha|0\rangle+ p\beta|1\rangle+\gamma|2\rangle\right)\\
 &=(\alpha^2+p\beta^2)|0\rangle+ \gamma_1|1\rangle+\gamma_2|2\rangle\\
 &=\gamma_1|1\rangle+\gamma_2|2\rangle,  \label{e2}
 \end{align}
 where $\gamma$, $\gamma_1$ and $\gamma_2$ are  amplitudes that we do not need to specify  exactly.

Since the amplitude of  $|0\rangle$ in the above quantum state $|q\rangle $ is 0, we  get the exact result after the measurement of $\gamma_1|1\rangle+\gamma_2|2\rangle$ in the standard basis $\{|0\rangle,|1\rangle,|2\rangle\}$.
\end{proof}

\begin{remark}
Similar matrices as $U_a$ have been first used in \cite{Amb98} and also in \cite{Amb09,AmYa11,Ber05,Yak10,Zhg12,Zhg13}. Our proof has been inspired by the methods introduced in \cite{Amb13}.   Similar methods can be found also in \cite{GDZ14}.
\end{remark}
\begin{remark}\label{Rm-2}
Let $\theta=\frac{2\pi}{N}$. If $l<\frac{N}{4}$ or $l>\frac{3N}{4}$, then $p=cos l\theta>0$ and Equations (\ref{E-eqs}) have no solution.
\end{remark}

In order to follow the same results
for the cases that   $l<\frac{N}{4}$ or $l>\frac{3N}{4}$,  we can use some  integer $q$ such that $\theta=\frac{q\cdot 2\pi}{N}$ and $ cos l\theta\leq0$. That is to find out some integer $q$ such that
\begin{equation}\label{e-1}
2\pi i+\frac{\pi}{2}\leq \frac{q 2\pi l}{N}\leq 2\pi i+\frac{3\pi}{2},
\end{equation}
where $i$ is an  integer. Notice that  Equality (\ref{e-1}) holds if and only if
\begin{equation}
\frac{N}{l}(i+\frac{1}{4})\leq q\leq \frac{N}{l}(i+\frac{3}{4}).
\end{equation}

If $l<\frac{N}{4}$, then $\frac{N}{l}> 4$ and
\begin{equation}
\frac{N}{l}(i+\frac{3}{4})-\frac{N}{l}(i+\frac{1}{4})=\frac{N}{l}\times\frac{1}{2}\geq 4\times \frac{1}{2}=2.
\end{equation}
Therefore, it is easy to find out integers $i$ and $q$, say $i=0$ and $q=\lceil \frac{N}{4l}\rceil$, such that  Equality (\ref{e-1}) holds.

For the case that $l>\frac{3N}{4}$, we need the following lemma.
\begin{lemma}
There exist  integers $i$ and $q$ such that
$\frac{N}{l}(i+\frac{1}{4})\leq q\leq \frac{N}{l}(i+\frac{3}{4})$
 for  $N>l>\frac{3N}{4}$.
\end{lemma}
\begin{proof}
Since $N>l>\frac{3N}{4}$,  we have $0<\frac{N-l}{l}<\frac{1}{3}$ and
\begin{equation}
\frac{N}{l}\left(i+\frac{1}{4}\right)=\left(1+\frac{N-l}{l}\right)\left(i+\frac{1}{4}\right)=i+\frac{1}{4}+i\frac{N-l}{l}+\frac{N-l}{4l}.
\end{equation}
Obviously, we have
\begin{equation}\label{E-1}
\frac{1}{4}<\frac{1}{4}+\frac{N-l}{4l}<\frac{1}{4}+\frac{1}{12}=\frac{1}{3}.
\end{equation}
Since $\frac{2}{3}-\frac{1}{4}>\frac{1}{3}$ and $\frac{N-l}{l}<\frac{1}{3}$, there must exist an integer, say $j$, such that
\begin{equation}\label{E-2}
\frac{1}{4}<j\frac{N-l}{l}-\left\lfloor j\frac{N-l}{l} \right\rfloor <\frac{2}{3}.
\end{equation}
Let $a=\left\lfloor \frac{N}{l}\left(j+\frac{1}{4}\right)\right\rfloor$. According to Inequalities (\ref{E-1}) and (\ref{E-2}), we have
\begin{equation}
\frac{1}{2}<\frac{N}{l}\left(j+\frac{1}{4}\right)-a= \frac{1}{4}+\frac{N-l}{4l}+j\frac{N-l}{l}-\left\lfloor j\frac{N-l}{l} \right\rfloor <1
\end{equation}
and
\begin{equation}
\frac{N}{l}\left(j+\frac{3}{4}\right)= \frac{N}{l}\left(j+\frac{1}{4}\right)+\frac{N}{l}\times\frac{1}{2}>a+\frac{1}{2}+\frac{1}{2}=a+1.
\end{equation}
Thus,  $q=a+1$ is the integer such that $\frac{N}{l}(j+\frac{1}{4})\leq q\leq \frac{N}{l}(j+\frac{3}{4})$
 for  $N>l>\frac{3N}{4}$.
\end{proof}

 Now we  can prove the following general result:
\begin{theorem}\label{Th-3}
For any fixed positive integers $N$ and $l$ such that $0<l<N$, the promise problem $A^{N,\,l}$ can be solved exactly by a 3 quantum basis states MOQFA ${\cal M}_{N,\,l}$.
\end{theorem}
\begin{proof}
We will proceed similarly as in the proof of Lemma \ref{Th-1}.
The choice of  $\theta$ depends on $N$ and $l$.  For given $N$ and $l$, if $\frac{N}{4}\leq l\leq \frac{3N}{4}$,  we  choose $\theta=\frac{2\pi}{N}$; if $l<\frac{N}{4}$, then we  choose $\theta=\frac{2\pi}{N}\lceil \frac{N}{4l}\rceil$; if $l>\frac{3N}{4}$,  we should first find out an integer $j$ such that $\frac{1}{4}<j\frac{N-l}{l}-\left\lfloor j\frac{N-l}{l} \right\rfloor <\frac{2}{3}$, then we  choose $q=\left\lfloor \frac{N}{l}\left(j+\frac{1}{4}\right)\right\rfloor+1$ and $\theta=\frac{q2\pi}{N}$. The rest of the proof is now similar to the one  of  Lemma \ref{Th-1}.
\end{proof}

\begin{corollary}\label{Th-7}
The promise problem $A^{N,r_1,r_2}$ can be solved exactly by a 3 quantum basis states MOQFA ${\cal M}_{N,r_1,r_2}$.
\end{corollary}
\begin{proof}
  Let $l=(r_2-r_1) \ {\it mod}\ N$. We will choose $\theta$ according to $N$ and $l$ in the same way as we did in the proof of Theorem \ref{Th-3}. Let
  ${\cal M}_{N,r_1,r_2}=(Q,\Sigma,\{U'_{\sigma}\,|\, \sigma\in\Sigma' \},|{0}\rangle,Q_a)$, where $Q=\{|0\rangle,|1\rangle,|2\rangle\}$, $Q_a=\{|0\rangle\}$,
$ U'_{|\hspace{-1mm}c}=(U_a)^{N-r_1{\it mod}\ N}U_{|\hspace{-1mm}c},\ U'_a=U_a$ and $U'_{\$}=U_{\$}$ ($U_{|\hspace{-1mm}c}$,  $U_a$, and $U_{\$}$ are the ones defined in Theorem \ref{Th-1}). The remaining part of the proof is an analogue of the one of Theorem \ref{Th-1}.

\end{proof}

We now deal with the minimal DFA for the promise problem $A^{N,\,l}$.
\begin{theorem}\label{Th-6}
For any fixed positive integers $N$ and $l$,  the minimal DFA solving the promise problem $A^{N,\,l}$ has $d$ states, where $d$ is the  smallest integer such that $d\mid N$ and $d\nmid l$.
\end{theorem}
\begin{proof}
Let $d$ be the  smallest integer such that $d\nmid (pN+l)$  for any integer $p$.
We consider now a $d$-state DFA ${\cal A}=(S,\Sigma,\delta,s_0,F)$, with the set of states   $S=\{s_0,s_1,\ldots,s_{d-1}\}$,  the set of accepting states  $F=\{s_{iN \ {\it mod}\  d}\,|\,i\geq 0\}$, and the transition function $\delta(s_i,a)=s_{(i+1)\ {\it mod}\  d}$.

If $x\in A_{yes}$, then $x=a^{iN}$ for some $i$ and $\widehat{\delta}(s_0,a^{iN})=s_{iN \ {\it mod}\  d}\in F$.

If $x\in A_{no}$, then $x=a^{jN+l}$ for some $j$.
We prove that for any $j\geq 0$, $\widehat{\delta}(s_0,a^{jN+l})=  s_{(jN+l)\ {\it mod}\  d}\not\in F$ by contradiction as follows:
We assume that $s_{(jN+l)\ {\it mod}\  d}\in F$.  Since $F=\{s_{iN \ {\it mod}\  d}\,|\,i\geq 0\}$,  then there  exists $i$ such that $s_{(jN+l)\ {\it mod}\  d}=s_{iN \ {\it mod}\  d}$,
 i.e. $(jN+l\equiv  iN \ {\it mod}\  d)$. Therefore, $(j-i)N+l\equiv 0 \ {\it mod}\  d$. Let $p=j-i$, then $d$ can  divide $pN+l$,   which is a contradiction.

Therefore the promise problem $A^{N,\,l}$ can be solved by a d-state DFA ${\cal A}$.

 \begin{figure}[htbp]
 \centering\epsfig{figure=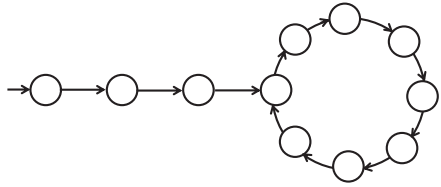,width=60mm}\\
 \centering\caption{The shape of the transition diagram of DFA solving  an  infinite unary promise problem.}\label{f1}
\end{figure}

Now we prove by a contradiction  that any DFA solving the promise problem $A^{N,\,l}$ has at least $d$ states.
Assume that there is an $m$-state DFA ${\cal A}'$ solving the promise problem $A^{N,\,l}$ and $m<d$.
Since both $A^{N,\,l}_{yes}$ and $A^{N,\,l}_{no}$ contain infinitely many unary strings, the shape of the transitions in the DFA ${\cal A}'$ must be like that in Figure \ref{f1}.

 Suppose now that  there are $k$ $(<d)$ states   before ${\cal A}'$ enters the cycle and the cycle has $t$ ($<d$)  states, says $s_0,\ldots,s_{t-1}$, such that $\delta(s_i,a)=s_{(i+1)\ {\it mod}\  t}$.
Since $t<d$, according to the assumption, these exits a $p$ such that $t\mid (pN+l)$. Therefore, we have  $((p+1)N+l-k\equiv N-k \ {\it mod}\  t)$, which means that
 $\widehat{\delta}(s_0, a^{(p+1)N+l})=\widehat{\delta}(s_0,a^N)$. This means that
 both $a^{(p+1)N+l}$ and $a^{N}$ are in $A^{N,\,l}_{yes}$ or in $A^{N,\,l}_{no}$  - a contradiction.

Therefore, the minimal DFA solving the promise problem $A^{N,\,l}$ has $d$ states.

Let $c$ be the smallest integer such that $c\mid N$ and $c\nmid l$.  We prove that $d=c$ as follows.

 Since $c\mid N$ and $c\nmid l$, we have $c\nmid (pN+l)$ for any integer $p$.  Since $d$ is the  smallest integer such that $d\nmid (pN+l)$  for any integer $p$,
  we have $c\geq d$.

 We now prove that $d\geq c$ as follows.

 Assume that $d\mid l$. We have $d\mid (dN+l)$, what contradicts to that $d$ is a integer such that $d\nmid (pN+l)$  for any integer $p$. Therefore $d\nmid l$.

  We are now going to prove that $d|N$.
  Let $g=\gcd(N,d)$. If $g=1$, then according to Euclid's theorem, then there must exist integers $u$ and $v$ such that $ud-vN=l$ and therefore $d\mid (vN+l)$, which is a contradiction. Hence $g>1$.

 Now let $d'=d/g$.   If  $g\mid l$, then $(d=d'g) \nmid (pN+l)$ for any $p$. Since $g\mid N$ and $g\mid l$, we have  $d' \nmid (pN+l)$ for any $p$. Therefore $d'<d$, what   contradicts to the assumption that $d$ is the  smallest integer such that $d\nmid (pN+l)$  for any integer $p$. Therefore, $g\nmid l$. Since $g\mid N$, we have  $g\nmid (pN+l)$ for any $p$. Now we have $g\geq d$. Since $g=\gcd(N,d)$, therefore $g=d$ and $d\mid N$.  Since $d\mid N$ and $d\nmid l$, therefore $d\geq c$.

\end{proof}

\begin{corollary}\label{Th-2}
For any fixed prime $N$ and fixed integer l such that $0< l<N$,  the minimal DFA solving the promise problem $A^{N,\,l}$ has $N$ states.
\end{corollary}

\begin{proof}
Since $N$ is a prime,  $d=N$ is the smallest integer  such that $d\mid N$ and $d\nmid l$.\ \ \

\end{proof}

\begin{corollary}
If $N$ is not a prime number and $l$ is a positive integer such that $\gcd(N,\,l)=1$,  then the minimal DFA solving the promise problem $A^{N,\,l}$ has $d$ states, where $d$ is the smallest integer such that $d \mid N$ and $d\neq 1$.
\end{corollary}
\begin{proof}
Assume that $d\neq 1 $ is the smallest
  integer such that $d \mid N$. We prove that $d\nmid l$ by a contradiction.
If  $d\mid l$,  then $\gcd(N,l)\geq d>1$, which is a contradiction.
 Therefore $d$ has to be  the  smallest integer such that $d\mid N$ and $d\nmid l$.

\end{proof}

\begin{remark}
If we choose $N=2^{k+1}$ and $l=2^k$, then $A^{N,\,l}$  is the promise problem studied in \cite{AmYa11}. In such a case   $d=2^{k+1}$.
If we choose $N=2^{k+1}(2m+1)$  and $l=2^{k}(2m+1)$, then $A^{N,\,l}$  is the promise problem mentioned in Section 3 in   \cite{AmYa11} and we  have $d=2^{k+1}$.
 If we choose $N=2n$ and $l=n$, where $n$ is odd, then $A^{N,\,l}$ is the promise problem mentioned in Section 4 in   \cite{AmYa11}. We have therefore in this case that  $pN+l=(2p+1)n$ is an odd integer and therefore $d=2$.
\end{remark}

\begin{remark}
Let $l=(r_2-r_1)\ {\it mod}\  N$. Then the size of the minimal  DFA for the promise problem $A^{N,r_1,r_2}$ is the same as the size of  the minimal DFA  for the promise problem $A^{N,\,l}$ and the proof is similar to the one of Theorem \ref{Th-6}.
\end{remark}

\section{Binary promise problems}

We  consider now a simple   binary promise problem  $B^{l}=(B_{yes}^{l},B_{no}^{l})$ with $B_{yes}^{l}=\{a^ib^{i}\mid i\geq 0\}$ and $B_{no}^{l}=\{a^ib^{i+l}\mid i\geq 0\}$, where $l$ is a fix positive number.
  Clearly, both $B_{yes}^{l}$ and $B_{no}^{l}$ are {\em nonregular languages}. However, we will prove that the promise problem $B^{l}=(B_{yes}^{l},B_{no}^{l})$ can be solved by an exact MOQFA and also by a DFA. See \cite{GY14} for more facts on classical automata solving promise problems.

 \begin{theorem}\label{Th-bi}
The promise problem $B^{l}$ can be solved exactly by a 2 quantum  basis states MOQFA ${\cal M}_{l}$.
\end{theorem}
\begin{proof}
Let $\theta=\frac{\pi}{2l}$ and  ${\cal M}_{l}=(Q,\Sigma,\{U_{\sigma}\,|\, \sigma\in\Sigma' \},|{0}\rangle,Q_a)$, where $Q=\{|0\rangle,|1\rangle\}$, $Q_a=\{|0\rangle\}$,
\begin{equation}
    U_{a}=\left(
      \begin{array}{cc}
        \cos \theta & \sin \theta   \\
        -\sin \theta  &\cos \theta\\
      \end{array}
    \right),\  U_{b}=\left(
      \begin{array}{cc}
        \cos \theta & -\sin \theta   \\
        \sin \theta  &\cos \theta\\
      \end{array}
    \right),
\end{equation}
and $U_{\$}=U_{|\hspace{-1mm}c}=I$.

If the input  $x\in B_{yes}^{l}$, then the quantum state before the measurement is
\begin{align}
    |q\rangle&=U_{\$}(U_b)^{i}(U_a)^{i}U_{|\hspace{-1mm}c}|0\rangle=\left(
      \begin{array}{cc}
        \cos \theta & -\sin \theta   \\
        \sin \theta  &\cos \theta\\
      \end{array}
    \right)^i\left(
      \begin{array}{cc}
        \cos \theta & \sin \theta   \\
        -\sin \theta  &\cos \theta\\
      \end{array}
    \right)^i|0\rangle\\
    &=\left(
      \begin{array}{cc}
        1 & 0 \\
        0 &1  \\
      \end{array}
    \right)|0\rangle=|0\rangle.
\end{align}

If the input   $x\in B_{no}^{l}$, then the quantum state before the measurement is
\begin{align}
|q\rangle&=U_{\$}(U_b)^{i+l}(U_a)^{i}U_{|\hspace{-1mm}c}|0\rangle=\left(
      \begin{array}{cc}
        \cos \theta & -\sin \theta   \\
        \sin \theta  &\cos \theta\\
      \end{array}
    \right)^{i+l}\left(
      \begin{array}{cc}
        \cos \theta & \sin \theta   \\
        -\sin \theta  &\cos \theta\\
      \end{array}
    \right)^i|0\rangle\\
    &=\left(
      \begin{array}{cc}
        \cos \theta & -\sin \theta   \\
        \sin \theta  &\cos \theta\\
      \end{array}
    \right)^{l}|0\rangle=\left(
      \begin{array}{cc}
        \cos l\theta & -\sin l\theta   \\
        \sin l\theta  &\cos l\theta\\
      \end{array}
    \right)|0\rangle\\
    &=\left(
      \begin{array}{cc}
        \cos \pi/2 & -\sin \pi/2  \\
        \sin \pi/2  &\cos \pi/2\\
      \end{array}
    \right)|0\rangle=\left(
      \begin{array}{cc}
        0 & -1  \\
        1  &0\\
      \end{array}
    \right)\left(
      \begin{array}{c}
        1   \\
        0\\
      \end{array}
    \right)=\left(
      \begin{array}{c}
        0   \\
        1\\
      \end{array}
    \right)=|1\rangle.
 \end{align}

Therefore we can get the exact result after the measurement in the standard basis $\{|0\rangle,|1\rangle\}$.

\end{proof}

 \begin{theorem}\label{Th-9}
  For any fixed $l$, the  minimal DFA solving the promise problem $B^{l}$ has $d$ states, where $d$ is the smallest integer such that $d\nmid l$.
\end{theorem}
\begin{figure}[htbp]
 \centering\epsfig{figure=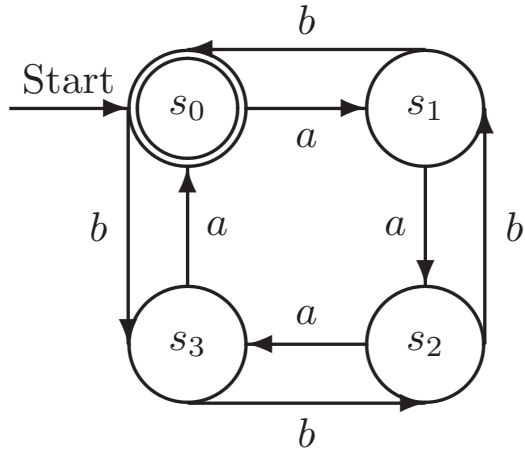,width=80mm}\\
 \centering\caption{DFA ${\cal A}$ when $d=4$.}\label{f2}
\end{figure}

\begin{proof}
Let us consider a $d$-state DFA ${\cal A}=(S,\Sigma,\delta,s_0,F)$, where $S=\{s_i|0\leq i<d\}$, $F=\{s_0\}$ and the transition function is defined as follows:
\begin{enumerate}
  \item $\delta(s_i,a)=s_{(i+1)\,{\it mod}\,d}$ and
  \item $\delta(s_i,b)=s_{(i-1)\,{\it mod}\,d}$.
\end{enumerate}

For example, if $d=4$, then the corresponding DFA ${\cal A}$ is as shown in Figure \ref{f2}.
It is easy to verify that the language recognized by DFA ${\cal A}$ is  $L=\{x\in\{a,b\}^*\mid \#_a(x)\equiv \#_b(x)\,{\it mod}\,d\}$, where $\#_a(x)$ ($\#_b(x)$) is the number of symbols $a$s ($b$s) in $x$.

Let $n=\#_a(x)$ and $m=\#_b(x)$. If the input string  $x\in B_{yes}^{l}$, then $n=m$. Therefore, we have $n\equiv m\,{\it mod}\,d$ and the DFA ${\cal A}$ accepts the input.

 If the input string  $x\in B_{no}^{l}$, then $m=n+l$.  Since $d\nmid l$, we have $n\not\equiv n+l \,{\it mod}\,d$ and the DFA ${\cal A}$ rejects the input.

 Therefore, the promise problem  $B^{l}$ can be solved by the DFA ${\cal A}$ and  its corresponding minimal DFA has no more than $d$ states.

 Now suppose that there is a $c$-state minimal DFA ${\cal A}'$ solving the promise problem  $B^{l}$ and $c<d$.
 For sufficient large $i$, after reading $a^*b^i$, ${\cal A}'$  will enter a cycle. Now suppose that the cycle has $t$ states, says $r_0,\ldots,r_{t-1}$ such that $\delta(r_j,b)=r_{(j+1)\,{\it mod}\,t}$. Since $t\leq c<d$, then we have $t\mid l$.  Let $r_k=\widehat{\delta}(s_0,a^ib^i)$. We have
 \begin{equation}
  \widehat{\delta}(s_0,a^ib^{i+l})=\widehat{\delta}(r_k,b^{l})=r_k,
 \end{equation}
 what means that the DFA ${\cal A}'$ accepts (or rejects) both $a^ib^i$ and $a^ib^{i+l}$ -- a contradiction. Therefore, $c\geq d$ and  the theorem has been proved.

\end{proof}

We consider now a more general promise problem $B^{N,\,l}=(B_{yes}^{N,\,l},B_{no}^{N,\,l})$ with $B_{yes}^{N,\,l}=\{a^ib^{i}\mid i\geq 0\}$ and $B_{no}^{N,\,l}=\{a^ib^{i+jN+l}\mid i,j\geq 0\}$, where $N$ and $l$ are fixed nonnegative integers such that  $0<l<N$.

 \begin{theorem}\label{Th-10}
The promise problem $B^{N,\,l}$ can be solved exactly by a 3 quantum basis states MOQFA ${\cal M}_{N,\,l}$.
\end{theorem}

\begin{proof}
We choose $\theta$  depending on  $N$ and $l$ as follows (we use the same strategy as  in Theorem \ref{Th-3}) :
\begin{enumerate}
  \item If $l<\frac{N}{4}$, then  $\theta=\frac{2\pi}{N}\lceil \frac{N}{4l}\rceil$.
  \item If  $\frac{N}{4}\leq l\leq \frac{3N}{4}$, then  $\theta=\frac{2\pi}{N}$.
  \item If $l>\frac{3N}{4}$,  we  first find out an integer $j$ such that $\frac{1}{4}<j\frac{N-l}{l}-\left\lfloor j\frac{N-l}{l} \right\rfloor <\frac{2}{3}$, and then we  choose $p=\left\lfloor \frac{N}{l}\left(j+\frac{1}{4}\right)\right\rfloor+1$ and $\theta=\frac{p2\pi}{N}$.
\end{enumerate}

We can design an MOQFA ${\cal A}$ with two quantum states as the one in Theorem \ref{Th-bi}. It is easy to see that  ${\cal A}$ $p$-solves the promise problem $B^{N,\,l}$, where $p=\cos l\theta\leq 0$.
According to Equality (\ref{e-p}),
we have $\alpha=\sqrt{\frac{-p}{1-p}}=\sqrt{\frac{-\cos l\theta}{1-\cos l\theta}}$ and $\beta=\sqrt{\frac{1}{1-p}}=\sqrt{\frac{1}{1-\cos l\theta}}$.

Let ${\cal M}_{N,\,l}=(Q,\Sigma,\{U_{\sigma}\,|\, \sigma\in\Sigma' \},|{0}\rangle,Q_a)$, where $Q=\{|0\rangle,|1\rangle,|2\rangle\}$, $Q_a=\{|0\rangle\}$,
\begin{equation}
    U_{|\hspace{-1mm}c}=\left(
      \begin{array}{ccc}
        \alpha \ & -\beta \ & 0  \\
        \beta \ & \alpha \ & 0   \\
        0 \ &  0  \  &1\\
      \end{array}
    \right), \
    U_{a}=\left(
      \begin{array}{ccc}
        1 & 0 & 0  \\
        0 & \cos \theta & \sin \theta   \\
        0 &  -\sin \theta  &\cos \theta\\
      \end{array}
    \right),\   U_{b}=\left(
      \begin{array}{ccc}
        1 & 0 & 0  \\
        0 & \cos \theta & -\sin \theta   \\
        0 &  \sin \theta  &\cos \theta\\
      \end{array}
    \right)
\end{equation}
and $U_{\$}=U_{|\hspace{-1mm}c}^{-1}$.

If the input  $x\in B_{yes}^{N,\,l}$, then the quantum state before the measurement is
\begin{align}
    |q\rangle&=U_{\$}(U_b)^{i}(U_a)^{i}U_{|\hspace{-1mm}c}|0\rangle=U_{\$}\left(
      \begin{array}{ccc}
        1 & 0 & 0  \\
        0 & \cos \theta & -\sin \theta   \\
        0 &  \sin \theta  &\cos \theta\\
      \end{array}
    \right)^i\left(
      \begin{array}{ccc}
        1 & 0 & 0  \\
        0 & \cos \theta & \sin \theta   \\
        0 &  -\sin \theta  &\cos \theta\\
      \end{array}
    \right)^iU_{|\hspace{-1mm}c}|0\rangle\\
    &=U_{\$}U_{|\hspace{-1mm}c}|0\rangle=|0\rangle.
\end{align}

If the input   $x\in B_{no}^{N,\,l}$, then the quantum state before the measurement is
\begin{align}
|q\rangle&=U_{\$}(U_b)^{i+jN+l}(U_a)^{i}U_{|\hspace{-1mm}c}|0\rangle=U_{\$}(U_b)^{jN+l}U_{|\hspace{-1mm}c}|0\rangle=U_{\$}(U_b)^{l}U_{|\hspace{-1mm}c}|0\rangle\\
&=\gamma_1|1\rangle+\gamma_2|2\rangle,  \ \ \ \ \ \  \ \ \ \ \ \ [\mbox{ Eqs. (\ref{e1}) to (\ref{e2}) }]
 \end{align}
 where $\gamma_1$ and $\gamma_2$ are  amplitudes that we do not need to specify  exactly.

Since the amplitude of $|0\rangle$ in the quantum state $|q\rangle $ is 0, we  get always the exact result after the measurement in the standard basis $\{|0\rangle,|1\rangle,|2\rangle\}$.
\end{proof}

 \begin{theorem}
   For any fixed $l$, the  minimal DFA solving the promise problem $B^{N,\,l}$ has $d$ states, where $d$ is the  smallest integer such that $d\mid N$ and $d\nmid l$.
\end{theorem}

\begin{proof}
Let $d$ be the  smallest integer such that $d\nmid (pN+l)$  for any integer $p$.
We consider a minimal DFA ${\cal A}=(S,\Sigma,\delta,s_0,F)$ accepting the language $L=\{x\in\{a,b\}^*\mid \#_a(x)\equiv \#_b(x)\,{\it mod}\,d\}$.

Let $n=\#_a(x)$ and $m=\#_b(x)$.  If the input $x\in B_{yes}^{N,\,l}$, then $n=m$. Therefore, we have $n\equiv m\,{\it mod}\,d$ and  ${\cal A}$ accepts the input.
If the input string  $x\in B_{no}^{N,\,l}$, then $m=n+pN+l$.  Since $d\nmid (pN+l)$, we have $n\not\equiv n+pN+l \,{\it mod}\,d$ and   ${\cal A}$ rejects the input.
 According to the proof of Theorem \ref{Th-9}, ${\cal A}$ has no more than $d$ states.

  Now suppose that there is a $c$-state minimal DFA ${\cal A}'$ solving the promise problem  $B^{N,\,l}$ and $c<d$.
 For sufficiently big  $i$, after reading $a^*b^i$, ${\cal A}'$   enters a cycle. Now suppose that the cycle has $t$ states, says $r_0,\ldots,r_{t-1}$ such that $\delta(r_j,b)=r_{(j+1)\,{\it mod}\,t}$. Since $t\leq c<d$,  there exists an integer $p$ such that $t\mid pN+l$.  Let $r_k=\widehat{\delta}(s_0,a^ib^i)$. We have
 \begin{equation}
  \widehat{\delta}(s_0,a^ib^{i+pN+l})=\widehat{\delta}(r_k,b^{pN+l})=r_k,
 \end{equation}
 which means that ${\cal A}'$ accepts (or rejects) both $a^ib^i$ and $a^ib^{i+Np+l}$ -- a contradiction. Therefore, $c\geq d$ and  the  minimal DFA solving the promise problem $B^{N,\,l}$ has to have $d$ states.

 According to the proof of  Theorem \ref{Th-6}, $d$ is  the  smallest integer such that $d\mid N$ and $d\nmid l$.
\end{proof}

\begin{remark}
Obviously, if $N$ is a prime, then $d=N$.  If $\gcd(N,\,l)=1$, then $d$ is the smallest integer  greater than 1 that divides $N$.
\end{remark}

\section{Conclusion and discussion}

Ambainis and Yakary{\i}lmz \cite{AmYa11} presented a family of promise problem,  i.e
\begin{equation}
    \left\{A^{N,\,l}=(A_{yes}^{N,\,l}=\{a^{iN}\},A_{no}^{N,\,l}=\{a^{iN+l}\}\,|\,\ N=2^{k+1}, l=2^k, i>0,\linebreak[0] \mbox{\ and \ } k>0 )\right\},
\end{equation}
and they proved that each promise problem can be solved exactly by an MOQFA with 2 quantum basis states, whereas the sizes of the corresponding DFAs are at least $N$.
Based on the techniques given in Ambainis \cite{Amb13},  we have generalized the result  in this paper, i.e we have proved that $N$ and $l$ can be any fixed positive integers such that $0< l<N$. We have given an exact MOQFA with 3 quantum basis states and a minimal DFA for  the new promise problems. Moreover we have proved  some similar results on two families of  binary promise problems in this paper.

Finite automata and the other restricted power quantum computing models to solve promise problems were intensively studied in the last several years \cite{AGKY14,AmYa11,BMP14,GY14,GDZ14,Nak14,RY14,Zhg12,Zhg13,Zhg14}.
 The method presented in this paper may be helpful  in finding more exact quantum finite automata or algorithms for other promise problems. Using our method, it is not hard to generalize the results from \cite{AGKY14}.

In theorems \ref{Th-3}, \ref{Th-7} and \ref{Th-10}, the number of quantum  basis states used in the corresponding  MOQFA is three. Actually, these results are not optimal.   Since the quantum states we use are in $\mathbb{R}^3$, they can be simulated in $\mathbb{C}^2$ that is by a qubit  \cite{Amb02}. Therefore the results can be improved to 2 quantum  basis states providing the ultimately optimal outcomes. We can refer the reader to \cite{Amb02} for details of such a technique.

The sizes of probabilistic finite automata and two-way nondeterministic finite automata   for the promise problem $A^k$ have been  studied in  \cite{GY14,RY14}.
In this paper we have only shown minimal DFAs for the promise problems introduced in this paper. It is not hard to determine the minimum amount of states required by bounded-error probabilistic finite automata solving the promise problems. We refer the reader to  \cite{RY14} for the proof technique.
It is also possible to determine the minimum amount of states required by two-way nondeterministic finite automata solving the promise problems. The reader can check \cite{GY14} for the proof technique. Actually, the results on the promise problem $A^{N,\,l}$ have been given recently in \cite{BMP14}.

Finally, we give a  problem for future research.
 Most of the more than polynomial speed-up (space efficient) results for exact quantum computing hold only in the case that we choose  special structures for ``yes" inputs and ``no" inputs.  For example, in this paper, we choose $A_{yes}^{N,\,l}=\{a^{iN}\,|\,\ i\geq 0\}$ and $A_{no}^{N,\,l}=\{a^{iN+l}\,|\,\ i\geq 0\}$, where $N$ and $l$ are  fixed positive integers.
      What if $l$ is not fixed? Says $n_1<l<n_2$, where $n_1<n_2<N$. Can we still get  similar exact quantum computing results?

\section*{Acknowledgements}
The authors are thankful to the anonymous referees   for their  careful reading, comments and suggestions that greatly helped to improve the results and the quality of the presentation.
We thank also Feidiao Yang and  Xiangfu Zou for their comments to original proof of Theorem \ref{Th-6}.
We thank also Carlo Mereghetti for sending us a copy of the paper \cite{BMP14}.
Gruska and Zheng were supported by
the Employment of Newly Graduated Doctors of Science for Scientific Excellence project/grant (CZ.1.07./2.3.00\linebreak[0]/30.0009)  of Czech Republic.
Qiu  was supported by the National
Natural Science Foundation of China (Nos. 61272058, 61073054).

\end{document}